\documentclass[10pt, twocolumn]{IEEEtran}

\usepackage{epsfig,latexsym}
\usepackage{amsmath}
\usepackage{amssymb}
\usepackage{subfigure}
\usepackage{hyperref}
\usepackage{cite}
\usepackage{indentfirst}
\usepackage{times}
\usepackage{fancyhdr}
\usepackage{lastpage}
\usepackage{bigints}
\usepackage{subfigure}
\usepackage{graphicx}
\usepackage{amsthm}

\newtheorem{lemma}{Lemma}
\newtheorem{remark}{Remark}

\usepackage[noend]{algpseudocode}
\usepackage{algorithmicx,algorithm}

\begin{document}
\title{A General Analytical Framework for the Resolution of Near-Field Beamforming}
\author{ Chenguang Rao, Zhiguo Ding, \emph{Fellow, IEEE}, Octavia A. Dobre, \emph{Fellow, IEEE}, and Xuchu Dai\thanks{
		The work is supported by the National Natural Science Foundation of China (No. 61971391). 
		
		C. Rao and X. Dai are with the CAS Key Laboratory of Wireless-Optical Communications, University of Science and Technology of China, No.96 Jinzhai Road, Hefei, Anhui Province, 230026, P. R. China. (e-mail: rcg1839@mail.ustc.edu.cn; daixc@ustc.edu.cn).
		
		Z. Ding is with Department of Electrical Engineering and Computer Science, Khalifa University, Abu Dhabi, UAE. (e-mail: zhiguo.ding@ku.ac.ae).
		
		O. A. Dobre is with the Department of Electrical and Computer Engineering, Memorial University, St. John’s, NL A1B 3X5, Canada. (e-mail: odobre@mun.ca).
	}\vspace{-3.2em}}
\maketitle
\begin{abstract}
	The resolution is an important performance metric of near-field communication networks. In particular, the resolution of near field beamforming measures how effectively users can be distinguished in the distance-angle domain, which is one of the most significant features of near-field communications. In a comparison, conventional far-field beamforming can distinguish users in the angle domain only, which means that near-field communication yields the full utilization of user spatial resources to improve spectrum efficiency. In the literature of near-field communications, there have been a few studies on whether the resolution of near-field beamforming is perfect. However, each of the existing results suffers its own limitations, e.g., each is accurate for special cases only, and cannot precisely and comprehensively characterize the resolution. In this letter, a general analytical framework is developed to evaluate the resolution of near-field beamforming. Based on this derived expression, the impacts of parameters on the resolution are investigated, which can shed light on the design of the near-field communications, including the designs of beamforming and multiple access tequniques.
\end{abstract}


\vspace{-1em}

\section{Introduction}

Recently, there have been an increasing number of studies on near-field communications \cite{NFC0,NFC1,NFC2}. One of the most significant studies is about the resolution of near-field beamforming, which is a performance metric to measure the effectiveness of distinguishing different users, an important feature of near-field communications. In traditional far-field communications, the distinction among different users relies on angular differences only. However, in near-field communications, thanks to the characteristics of the used spherical wave channel model, the differences in users' distances can also be used for user distinction \cite{beamforming1,beamforming2}. When the resolution is perfect, the users in different locations can be perfectly distinguished, such that user spatial resources can be fully utilized to improve spectrum efficiency. In \cite{Dai}, the resolution was expressed via the Fresnel functions, and was proved to be perfect when the number of base station antennas approaches infinity. Based on this property of resolution, the authors introduced the location based multiple access (LDMA) scheme to improve the spectral efficiency and fully utilize the location information of users. However, in \cite{Ding}, the author pointed out that unless users are close to the base station, the resolution of near-field beamforming is not perfect. Specifically, if the users' distances to the base station are proportional to the Rayleigh distance, the resolution remains poor even if the number of antennas approaches infinity. Based on this observation, the author introduced the non-orthogonal multiple access (NOMA) scheme for near-field communications, which can effectively exploit the imperfect resolution.

It is important to point out that the existing studies about the resolution of near-field beamforming have limitations, being accurate only in special cases. Except for these cases, there is a lack of an analytical framework to characterize the resolution precisely and comprehensively, which motivates our work. In this letter, a general analytical framework is developed to evaluate the resolution of near-field beamforming, where the obtained analytical results can be applied to the cases with a uniform planar array (UPA) and a uniform linar array (ULA). In particular, a concise closed-form expression for the resolution of near-field beamforming is introduced. Compared with the existing results, this expression can be evaluated without involving integration calculations and has a higher accuracy, which can facilitate studies related to near-field communications. Based on the derived expression of resolution, the impacts of various parameters are also investigated. Specifically, the conditions for achieving perfect and imperfect resolution are analyzed in this paper, which provides insights for more studies on resolution, such as the designs of beamforming and multiple access techniques. In addition, the case when the base station is equipped with a ULA is also analyzed as a special case. Finally, the analytical results are validated through the simulation results. 

\vspace{-1em}

\section{System Model}

Consider a communication network with a base (\(\mathrm{BS}\)), equipped with a UPA with \((2M+1)\times(2N+1)\) elements. The antenna spacing is denoted by \(d\). The rectangular coordinates of the UPA elements are given by \(\mathbf{s}_{m,n} = (nd,0,md)\), \(-M\leq m\leq M\), \(-N\leq n\leq N\). In this letter, \(d\) is set as \(d=\frac{\lambda}{2}\), where \(\lambda\) represents the carrier wavelength. The \(i\)-th user is denoted by \(\mathrm{U}_i\), whose spherical coordinate is given by \(\mathbf{r}_i = (r_i,\theta_i,\phi_i)\), with \(0<\theta_i,\phi_i<\pi\), and \(r_i\) assumed to be \(r_i\leq d_{Ray}\), where \(d_{Ray} = \frac{8d^2(M^2+N^2)}{\lambda} = 2\lambda(M^2+N^2)\) represents the Rayleigh distance. The channel vector between \(\mathrm{BS}\) and \(\mathrm{U}_i\) is modeled as \(\mathbf{h}_i^H = \beta_i\mathbf{b}_i^H\),
where \(\beta_i = \frac{\sqrt{t}\lambda}{4\pi r_i}\) is the channel gain, and
\begin{equation}\label{bi}
	\begin{small}
		\begin{aligned}
			\mathbf{b}_i =& \frac{1}{\sqrt{t}}
			\left[e^{-j\frac{2\pi}{\lambda}(||\mathbf{r}_i-\mathbf{s}_{-M,-N}||-r_i)},e^{-j\frac{2\pi}{\lambda}(||\mathbf{r}_i-\mathbf{s}_{-M,-N+1}||-r_i)},\right.\\&\dots\ ,e^{-j\frac{2\pi}{\lambda}(||\mathbf{r}_i-\mathbf{s}_{M,N}||-r_i)}\left]^T,\right.
		\end{aligned}
	\end{small}
\end{equation}
represents the response array, with \(t = (2M+1)(2N+1)\) as the normalization coefficient, and \(||\cdot||\) as the norm of a vector. By applying the Fresnel approximation in \cite{Fresnel1}, \(||\mathbf{r}_i-\mathbf{s}_{m,n}||\) can be expressed as follows\footnote{This approximation does not cause much loss of accuracy, which has been validated in many studies, such as \cite{Dai},\cite{Fresnel2}.}:
\begin{equation}\label{ri}
	\begin{small}
		\begin{aligned}
			||\mathbf{r}_i-\mathbf{s}_{m,n}|| &\approx r_i-nd\cos\theta_i\sin\phi_i+\frac{n^2d^2(1-\cos^2\theta_i\sin^2\phi_i)}{2r_i}\\
			&-md\cos\phi_i+\frac{m^2d^2\sin^2\phi_i}{2r_i}.
		\end{aligned}
	\end{small}
\end{equation}
The received message by \(\mathrm{U}_i\) can be expressed as follows:
\begin{equation}
	y_i = \mathbf{h}_i^H\mathbf{p}x+n_i,
\end{equation}
where \(\mathbf{p}\) represents the beamforming vector, and \(n_i\) is the white noise with power \(P_N\). Define the resolution of near-field beamforming as \(\Delta = |\mathbf{b}_1^H\mathbf{b}_2|^2\), which measures the effectiveness of distinguishing the near-field users. Specifically, if \(\Delta\approx 0\), the resolution is considered perfect. In this case, the response vectors for two different users are orthogonal, i.e., the near-field beamforming can focus the energy on a specific location, which facilitate the implementation of LDMA\cite{Dai}. On the other hand, if \(\Delta\approx 1\), the resolution is considered poor. In this case, different users share the same near-field beam, which facilitates the implementation of NOMA\cite{Ding}. Therefore, the characteristic of \(\Delta\) is significant for near-field communications, which will be discussed in the following section.

\vspace{-1em}

\section{Resolution of Near-Field Beamforming}
\subsection{Main Results}
In this section, a closed-form expression of \(\Delta\) is derived. Then the impacts of various parameters, i.e., \(\theta_i,\phi_i,r_i,M,N\), on \(\Delta\) are investigated.

For notational simplicity, some definitions are introduced as follows: \(a = \frac{1}{2}(\cos\phi_1-\cos\phi_2)\), \(b = \frac{1}{2}(\cos\theta_1\sin\phi_1-\cos\theta_2\sin\phi_2)\), \(c = \frac{\lambda}{8}\left(\frac{\sin^2\phi_1}{r_1}-\frac{\sin^2\phi_2}{r_2}\right)\), and \(z = \frac{\lambda}{8}\left(\frac{1-\cos^2\theta_1\sin^2\phi_1}{r_1}-\frac{1-\cos^2\theta_2\sin^2\phi_2}{r_2}\right)\).

\begin{lemma}\label{T1}
	\(\Delta\) can be expressed as follows:
	\begin{equation}\label{D2}
		\Delta \approx \frac{1}{t^2}(\sum_{i=1}^{3}I_i+\frac{I_2I_3}{t}),
	\end{equation}
	where \(I_1 = t\),
	\begin{equation}
		I_2 = 2(2M+1)\sum_{s = 1}^{2N}\Phi(z,s,N)\cos(2\pi bs),
	\end{equation}
	\begin{equation}
		I_3 = 2(2N+1)\sum_{r=1}^{2M}\Phi(c,r,M)\cos(2\pi ar),
	\end{equation}
	\begin{equation}\label{S}
		\Phi(x,y,K) = \frac{\sin(2\pi xy(2K-y+1))}{\sin(2\pi xy)}.
	\end{equation}
\end{lemma}
\begin{proof}
	See Appendix \ref{Appendix1}.
\end{proof}

Compared to \cite{Ding}, the high accuracy approximation of \(\Delta\) can be obtained from Lemma \ref{T1} for arbitrary choices of the users' distances to the \(\mathrm{BS}\). Compared with the result in \cite{Dai}, (\ref{D2}) does not require integral calculations and can be evaluated with a lower complexity. 

From (\ref{D2}), it can be observed that the value of \(\Delta\) depends on the parameters \(\theta_i,\phi_i,r_i,M,N\), and \(\lambda\). The asymptotic characteristic of \(\Delta\) when \(M,N\to +\infty\) is important, and will be discussed in the following.

\begin{remark}\label{R1}
	When \(\theta_1\neq\theta_2\) or \(\phi_1\neq\phi_2\), at least one of \(a\) and \(b\) is non-zero. By applying the inequality \(\frac{\sin nx}{\sin x}\leq n\), when \(b\neq 0\), the upper bound on \(I_2\) can be obtained as follows:
	\begin{equation}
		\begin{small}
			\begin{aligned}
				I_2&\leq 2(2M+1)\sum_{s=1}^{2N}(2N-s+1)\cos(2\pi bs)\\
				&=(2M+1)\frac{(2N+1)cos(2\pi b)-cos(2\pi b(2N+1))-2N}{2\sin^2(\pi b)}\\
				&\leq\frac{2M+1}{\sin^2(\pi b)}.
			\end{aligned}
		\end{small}
	\end{equation}
	Similarly, when \(a \neq 0\), the upper bound on \(I_3\) can be derived as \(I_3\leq\frac{2N+1}{\sin^2(\pi a)}\). Then, when both \(a\) and \(b\) are non-zero, \(\Delta\) can be known as follows:
	\begin{equation}
		\Delta\leq \frac{1+2M+2N}{t}+\max\{\frac{2M(2M+1)}{t^2\sin^2(\pi b)},\frac{2N(2N+1)}{t^2\sin^2(\pi a)}\}.
	\end{equation}
	When \(t = (2M+1)(2N+1)\to +\infty\), the right side of the inequality approaches \(0\), so \(\Delta\to 0\). This result indicates that the resolution of near-field beamforming becomes perfect in the angle domain, which is similar to the conclusion made in the far-field communication network. 
\end{remark}

\begin{remark}
	When \(\theta_1=\theta_2=\theta\) and \(\phi_1=\phi_2=\phi\), \(\cos(2\pi ar) = \cos(2\pi bs) = 0\). Notice that when \(2\pi xy(2K-y+1)\) is small (\(\leq0.25\)), the small angle approximation \(\sin x\approx x\) can be applied. In this case, \(\Phi(x,y,K)\approx 2K-y+1\), and
	\begin{equation}
		I_2 \approx 2(2M+1)\sum_{s = 1}^{2N}(2N-s+1) = 2Nt, I_3\approx 2Mt.
	\end{equation}
	Then \(\Delta\) can be approximated as follows:
	\begin{equation}
		\Delta \approx \frac{1}{t^2}(t+2Mt+2Nt+4MNt^2) = 1.
	\end{equation}
	This approximation requires the following conditions:
	\begin{equation}
		\begin{aligned}
			&2\pi zs(2N-s+1)\leq 0.25,2\pi cr(2N-r+1)\leq 0.25,\\
			&s = 1,\dots,N,r=1,\dots,M.
		\end{aligned}
	\end{equation}
	Since \(s(2N-s+1) \leq N(N+1)\), \(r(2M-r+1) \leq M(m+1)\), and the conditions can be solved as \(z\leq\frac{1}{8\pi N(N+1)}\), \(c\leq\frac{1}{8\pi M(M+1)}\). When \(r_1\approx r_2\), \(z,c\approx 0\), and in this case \(\Delta\approx 1\). When \(r_1\neq r_2\), with \(r = \min\{r_1,r_2\}\), \(c\) and \(z\) can be estimated respectively as follows:
	\begin{equation}
		\begin{aligned}
			c &= \frac{\sin^2\phi}{8}(\frac{1}{r_1}-\frac{1}{r_2})\lambda \sim \frac{\sin^2\phi}{8r}\lambda
		\end{aligned}
	\end{equation}
	\begin{equation}
		\begin{aligned}
			z &= \frac{1-\cos^2\theta\sin^2\phi}{\sin^2\phi}c \sim \frac{1-\cos^2\theta\sin^2\phi}{8r}\lambda.
		\end{aligned}
	\end{equation}
	Then, the conditions can be derived:
	\begin{equation}\label{Cr}
		r\geq \pi\lambda \max\{M(M+1)\sin^2\phi,N(N+1)(1-\cos^2\theta\sin^2\phi)\}.
	\end{equation}
	When this condition holds, \(\Delta\approx 1\), otherwise, the error of the small-angle approximation cannot be negligible. In this case, \(\Phi(x,y,K)< 2K-y+1\), and \(\Delta<1\). It is worth to point out that the expression in Lemma \ref{T1} still holds for this case with high accuracy. 
\end{remark} 

\begin{remark}\label{R3}
	In practice, the Rayleigh distance \(d_{Ray}\) increases as \(M,N\) increase, i.e., the range of near-field communication becomes larger. For this consideration, we denote \(r = \beta d_{Ray}\), \(0< \beta\leq 1\); then the condition (\ref{Cr}) can be expressed as follows:
	\begin{equation}\label{Cbeta}
		\begin{aligned}
			\beta\geq&\frac{\pi}{2(M^2+N^2)}\times\\
			&\max\{M(M+1)\sin^2\phi,N(N+1)(1-\cos^2\theta\sin^2\phi)\}.
		\end{aligned}
	\end{equation}
	Denote \(l=N/M\).p When \(N\to +\infty\), \(\frac{N(N+1)}{M^2+N^2}\to \frac{1}{l^2+1}\), \(\frac{M(M+1)}{M^2+N^2}\to \frac{l^2}{l^2+1}\). Then (\ref{Cbeta}) can be rewritten as
	\begin{equation}\label{Cbeta1}
		\beta\geq\frac{\pi}{2}\max\{\frac{l^2\sin^2\phi}{l^2+1},\frac{1-\cos^2\theta\sin^2\phi}{l^2+1}\}.
	\end{equation}
	When (\ref{Cbeta1}) holds, \(\Delta\approx 1\) even if \(M,N\to+\infty\). This result is consistent with the Corollary 1 in \cite{Ding}.
\end{remark}

\begin{remark}
	The analysis in Remark \ref{R3} focuses on the case where \(r_1\) and \(r_2\) increase as \(M\) and \(N\) increase. If \(r_1\) and \(r_2\) are constant, condition (\ref{Cr}) cannot hold as \(M\) and \(N\) increase continually, which can cause a reduction of \(\Delta\). When \(M,N\to+\infty\), \(\Delta\to 0\). This result is consistent with the conclusion in \cite{Dai}.
\end{remark}

\begin{remark}
	Summarizing the above discussion about \(\Delta\) when \(\theta_1=\theta_2=\theta\) and \(\phi_1=\phi_2=\phi\), we have
	\begin{itemize}
		\item Considering the impact of \(\beta = \frac{\min\{r_1,r_2\}}{d_{Ray}}\) on \(\Delta\), the conclusion is that a larger \(\beta\) leads to a larger \(\Delta\), until the inequality (\ref{Cbeta}) holds, where \(\Delta\approx 1\). In this case, even if \(M,N\to+\infty\), \(\Delta\) is still approximately \(1\). Conversely, when \(\beta\) decreases, D also decreases until it approaches zero.
		\item  Considering the impacts of \(M,N\) on \(\Delta\) when the users' locations are fixed, the conclusion is that if \(r_1\) and \(r_2\) are constant, \(\Delta\to 0\), when \(M,N\to+\infty\).
		\item  Note that except for the special cases mentioned above, \(\Delta\) does not take on values of \(0\) or \(1\); instead, it takes a value between zero and one. The specific value depends on the parameters and can be obtained from (\ref{D2}). The impacts of parameters on \(\Delta\) will be further discussed in Section \ref{simulation}.
	\end{itemize}
\end{remark}

\subsection{The Case of ULA}
When \(M=0,\phi_1 = \phi_2 = \frac{\pi}{2}\), the UPA is degraded to the ULA. The results for ULA are present in the following:
\begin{lemma}
	\(\Delta\) can be expressed as:
	\begin{equation}
			\begin{aligned}
				\Delta = \frac{1}{2N+1}+\frac{2}{(2N+1)^2}\sum_{s=1}^{2N}\Phi(z,s,N)\cos(2\pi bs).
			\end{aligned}
	\end{equation}
\end{lemma}
\begin{proof}
	By setting \(M=0\), the result can be obtained directly from Lemma \ref{T1}.
\end{proof}
The value of \(\Delta\) depends on the parameters \(\theta_i,r_i,N\). Similar results with the UPA case are presented as follows:
\begin{remark}
	If \(\theta_1\neq\theta_2\), then \(\Delta\to 0\) when \(N\to+\infty\).
\end{remark}
\begin{remark}
	When \(\theta_1=\theta_2=\theta\), denote \(r = \min\{r_1,r_2\} = kr_R\), and \(0<\beta\leq1\). Under this conditions, if \(\beta\geq \frac{\pi\sin^2\theta}{2}\), \(\Delta\approx 1\) when \(N\to +\infty\). This result agrees with the Corollary 1 in \cite{Ding}.
\end{remark}
\begin{remark}
	When \(\theta_1=\theta_2=\theta\), and if \(r_1\) and \(r_2\) are fixed, then \(\Delta\approx 0\) when \(N\to+\infty\). This result agrees with the Corollary 1 in \cite{Dai}.
\end{remark}

\vspace{-2em}

\section{Simulation Results}\label{simulation}

In this section, simulations are performed to verify the analytical results introduced in the previous section. The carrier wavelength is set as \(\lambda = 10^{-2}\) m, and the antenna spacing is chosen as \(d = \frac{\lambda}{2}=5\times10^{-3}\) m. \(N\) is set as \(N=128\). Each pair of users are located at the same angle, i.e, \(\theta_1=\theta_2=\theta\) and \(\phi_1=\phi_2=\phi\). 

Fig. \ref{figure_UPAk} presents the numerical and analytical results of \(\Delta\) in the case of UPA with different \(\beta,M,N\). This figure demonstrates the high accuracy of the analytical results in Lemma \ref{T1}. From the figure, it can be observed that when \(\beta\) is large, \(\Delta\approx 1\). When \(\beta\) decreases, \(\Delta\) also decreases, i.e., when users are closer to the \(\mathrm{BS}\), the resolution enhances. However, when \(\beta = 0.1\), the value of \(\Delta=0.4\) remains large. This indicates that within the majority of the near-field communication range, the resolution of near-field beamforming is not perfect. In addition, when \(N\) is fixed, a smaller \(M\) leads to a better resolution. It can be indicated that a larger difference between \(M\) and \(N\) leads to a better resolution.

\begin{figure}
	\centering
	\includegraphics[scale=0.6]{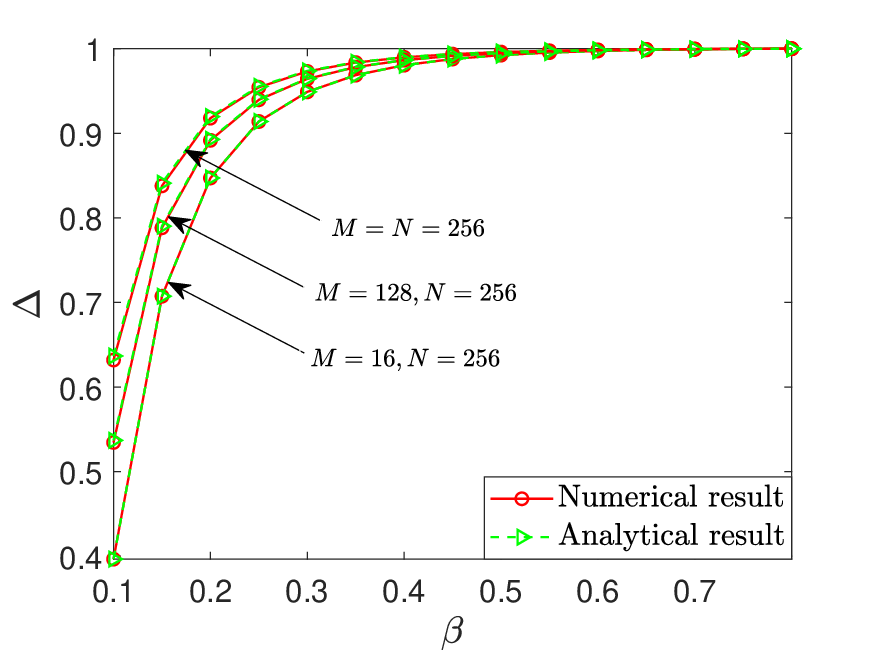}
	\caption{The results for \(\Delta\) in the case for UPA.} 
	\label{figure_UPAk}
\end{figure}

Fig. \ref{figure_ULAk} presents the results of \(\Delta\) in the case of ULA, where the impact of \(\beta\) on \(\Delta\) is investigated. For comparison, the approximated results in \cite{Ding} and \cite{Dai} are also presented. Compared with the other two approximated results, the analytical results from (\ref{D2}) have the best accuracy. In addition, similar to Fig. \ref{figure_UPAk}, when \(\beta\) is large, \(\Delta\approx 1\). As \(\beta\) decreases, \(\Delta\) decreases as well, until \(\beta\) reaches \(0.2\), at which point \(\Delta\) drops to about \(0.9\). This observation is consistent with the conclusion in Fig \ref{figure_UPAk}, i.e., within the majority of the near-field range, the near-field beamforming resolution is not perfect.

\begin{figure}
	\centering
	\includegraphics[scale=0.6]{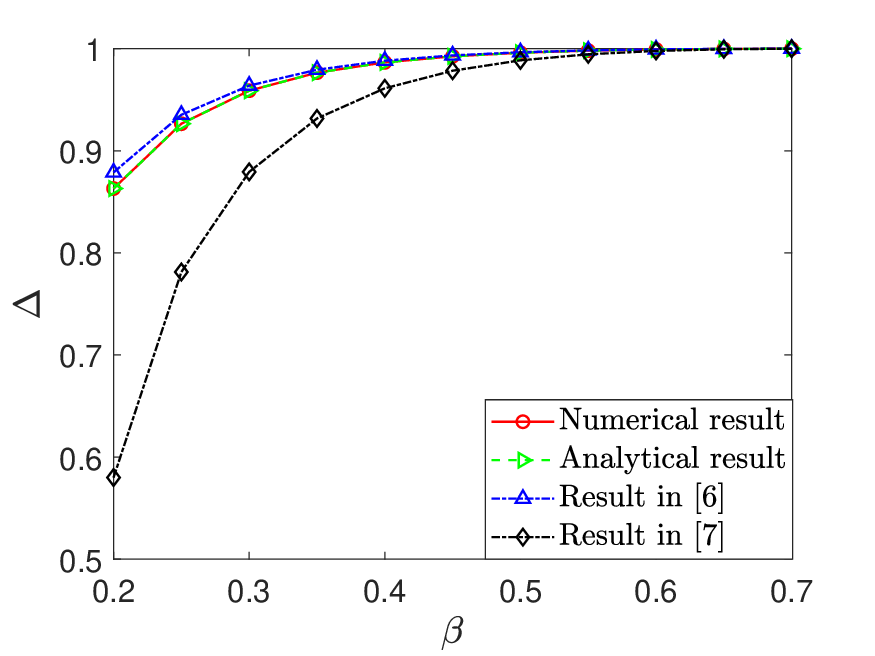}
	\caption{The results for \(\Delta\) in the case for ULA.}
	\label{figure_ULAk}
\end{figure}

In Fig. \ref{figure_ULAr}, the impact of \(r\) is investigated. The results in \cite{Dai} are also presented for comparison. As shown in the figure, both  approximated results have good accuracy, while the new analytical results presented in this letter are more accurate. It can be observed that as \(r_1\) increases, \(\Delta\) increases until \(\Delta \approx 1\), at which point \(r_1=r_2\). Then as \(r_1\) continues to increase, \(\Delta\) decreases until around \(0.1\). This indicates that when one user is very close to \(\mathrm{BS}\), the resolution is good if another user is not close. In addition, it can be seen that the resolution is poorer when \(r_2=40\) m, which indicates that the resolution for the case with users close to \(\mathrm{BS}\) is better. Therefore, to achieve a better resolution in near-field communication beamforming, it is necessary to ensure that at least one user is very close to \(\mathrm{BS}\) while setting a significant distance between the users.

\begin{figure}
	\centering
	\includegraphics[scale=0.6]{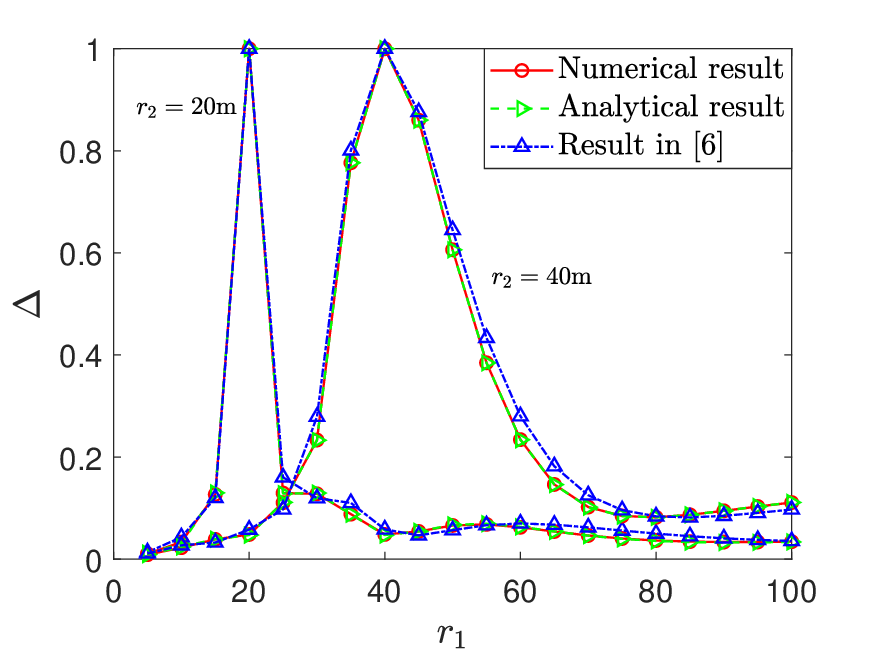}
	\caption{The results for \(\Delta\) in the case of ULA. \(r_2\) is fixed.}
	\label{figure_ULAr}
\end{figure}

To investigate the impact of \(r\) on \(\Delta\), Fig. \ref{figure_ULAr2} presents the results when \(r_2-r_1\) is constant. It can be observed that when \(r_1\) is very small, \(\Delta\approx 0\). As \(r_1\) increases, \(\Delta\) also increases. When \(r_1\) reaches \(100\) m, \(\Delta\) reaches around \(0.9\). This suggests that when two users are close to the base station, the resolution of near-field beamforming is good. On the other hand, if both users are far from \(\mathrm{BS}\), the resolution will deteriorate with increasing distances. Recall that the Rayleigh distance \(d_{Ray} \approx 1300\) m for this setup. When the \(r_1\) increases to about \(1/13\) of the Rayleigh distance, \(\Delta\) is almost \(1\), and the resolution at this point is poor. This indicates a conclusion similar to Fig. \ref{figure_UPAk} and Fig. \ref{figure_ULAk}, i.e., within the majority of the near-field communication distance, the resolution is not good. In addition, the figure shows that the resolution for the case with \(r_2-r_1=40\) m is better than the resolution for the case with \(r_2-r_1=20\) m. This is because the users which are further apart are easier to be distinguished.

\begin{figure}
	\centering
	\includegraphics[scale=0.6]{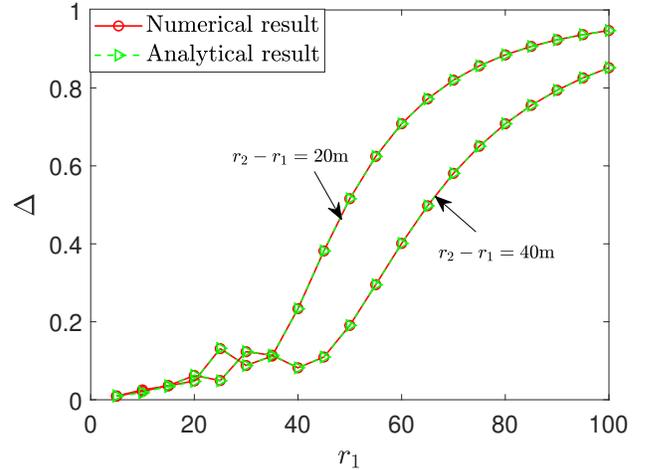}
	\caption{The results for \(\Delta\) in the case of ULA. \(r_2-r_1\) is constant.}
	\label{figure_ULAr2}
\end{figure}

\vspace{-1em}

\section{Conclusion}
In this letter, a concise, general, and high-accuracy closed-form expression of the resolution of the near-field beamforming \(\Delta\) has been obtained, and used to investigating the impacts of parameters on \(\Delta\). It has been shown that when \(\beta = \frac{\min\{r_1,r_2\}}{d_{Ray}}\) is large enough, \(\Delta\approx 1\)  even if \(M,N\to +\infty\). Another extreme case is that when two users' locations are fixed, \(\Delta\approx 0\) when \(M,N\to +\infty\). In most cases, \(0<\Delta<1\). In addition, the special case of ULA has also been analyzed. Finally, simulation results have been presented to verify the derived analytical results.

\appendices
\section{Proof of Theorem \ref{T1}}\label{Appendix1}
Define the function \(f(m,n)\) as follows:
\begin{equation}
	f(m,n) = -am-bn+cm^2+zn^2.
\end{equation}
Then \(\Delta\) can be written as follows:
\begin{equation}
	\Delta \approx \frac{1}{t^2}|\sum_{m=-M}^{M}\sum_{n=-N}^{N}\exp\{2j\pi f(m,n)\}|^2.
\end{equation}
Expanding the square of complex modulus, \(\Delta\) can be derived  as:
\begin{equation}
	\begin{aligned}
		\Delta &= \frac{1}{t^2}\sum_{m=-M}^{M}\sum_{n=-N}^{N}\sum_{p=-M}^{M}\sum_{q=-N}^{N}e^{2j\pi (f(p,q)-f(m,n))}\\
		&= \frac{1}{t^2}\sum_{m=-M}^{M}\sum_{n=-N}^{N}\sum_{r=-M-m}^{M-m}\sum_{s=-N-n}^{N-n}e^{2j\pi (f(m+r,n+s)-f(m,n))}\\
		&= \frac{1}{t^2}\sum_{r=-2M}^{2M}\sum_{s=-2N}^{2N}J,
	\end{aligned}
\end{equation}
where
\begin{equation}
	J = \sum_{-M-r\leq m\leq M-r\atop-M\leq m\leq M}\sum_{-N-s\leq n\leq N-s\atop-N\leq n\leq N}e^{2j\pi (f(m+r,n+s)-f(m,n))}.
\end{equation}
Since \(e^{2j\pi x}+e^{-2j\pi x} = 2\cos(2\pi x)\), \(\Delta\) can be expressed by the sum of multiple terms with different \(r\) and \(s\) as \(\Delta = \frac{1}{t^2}\sum_{i=1}^{5}I_i\), where \(I_i\) is expressed as follows:
\begin{itemize}
	\item \(I_1\): when \(r=s=0\).
	\begin{equation}
		I_1 = \sum_{m=-M}^{M}\sum_{n=-N}^{N}e^0 = (2M+1)(2N+1).
	\end{equation}
	\item \(I_2\): when \(r=0\) and \(s\neq 0\).
	\begin{equation}
		I_2 = \sum_{-2N\leq s\leq2N\atop s\neq 0}^{2N}\sum_{m=-M}^{m=M}\sum_{-N-s\leq n\leq N-s\atop-N\leq n\leq N}e^{2j\pi (f(m,n+s)-f(m,n))}.
	\end{equation}
	We can pair the terms of \((s_0,n_0)\) and \((-s_0,n_0+s_0)\), which leads to the following expression: 
	\begin{equation}
		\begin{aligned}
			I_2 
			&= \sum_{s = 1}^{2N}\sum_{m=-M}^{m=M}\sum_{n=-N}^{N-s}(e^{2j\pi(f(m,n+s)-f(m,n))}\\
			&\quad\quad\quad\quad\quad\quad\quad\quad\quad+e^{2j\pi(f(m,n)-f(m,n+s))})\\
			&= 2\sum_{s = 1}^{2N}\sum_{m=-M}^{m=M}\sum_{n=-N}^{N-s}\cos(2\pi(f(m,n+s)-f(m,n)))\\
			&= 2(2M+1)\sum_{s = 1}^{2N}\sum_{n=-N}^{N-s}\cos(2\pi(-bs+z(s^2+2ns)))\\
			& = 2(2M+1)\sum_{s = 1}^{2N}\Phi(z,s,N)\cos(2\pi bs).
		\end{aligned}
	\end{equation}
	\item \(I_3\): when \(s=0\) and \(r\neq 0\). This case is similar to \(I_2\). \(I_3\) can be found to be given as:
	\begin{equation}
		I_3 = 2(2N+1)\sum_{r=1}^{2M}\Phi(c,r,M)\cos(2\pi ar).
	\end{equation}
	\item \(I_4\): when \(s,r\geq 1\) or \(s,r\leq -1\). We can pair the term with \((r_0,s_0,m_0,n_0)\) and \((-r_0,-s_0,m_0+r_0,n_0+s_0)\), and obtain
	\begin{equation}
		\begin{small}
			\begin{aligned}
				I_4 &= \sum_{r=1}^{2M}\sum_{s=1}^{2N}\sum_{m=-M}^{M-r}\sum_{n=-N}^{N-s}(e^{2j\pi(f(m+r,n+s)-f(m,n))}\\
				&\quad\quad\quad\quad\quad\quad\quad\quad\quad+e^{2j\pi(f(m,n)-f(m+r,n+s))})\\
				&=2\sum_{r=1}^{2M}\sum_{s=1}^{2N}\Phi(z,s,N)s\\
				&\quad\quad\quad\quad\quad\times\sum_{m=-M}^{M-r}\cos(2\pi(-ar-bs+c(r^2+2mr)))\\
				&=2\sum_{r=1}^{2M}\sum_{s=1}^{2N}\Phi(z,s,N)\Phi(c,r,M)\cos(2\pi(ar+bs)).
			\end{aligned}
		\end{small}
	\end{equation}
	\item \(I_5\): when \(s\geq 1,r\leq -1\) or \(s\leq -1,r\geq 1\). This case is similar to \(I_4\). \(I_5\) can be shown to be given by:
	\begin{equation}
		\begin{small}
			\begin{aligned}
				I_5 &=
				\sum_{r=1}^{2M}\sum_{s=-2N}^{1}\sum_{m=-M}^{M-r}\sum_{n=-N-s}^{N}(e^{2j\pi(f(m+r,n+s)-f(m,n))}\\
				&\quad\quad\quad\quad\quad\quad\quad\quad\quad+e^{2j\pi(f(m,n)-f(m+r,n+s))})\\
				&=2\sum_{r=1}^{2M}\sum_{s=-2N}^{1}\Phi(z,-s,N)\Phi(c,r,M)\cos(2\pi(ar+bs))\\
				&=2\sum_{r=1}^{2M}\sum_{s=1}^{2N}\Phi(z,s,N)\Phi(c,r,M)\cos(2\pi(ar-bs)).
			\end{aligned}
		\end{small}
	\end{equation}
\end{itemize}
Since \(\cos(2\pi(ar+bs))+\cos(2\pi(ar-bs)) = 2\cos(2\pi ar)\cos(2\pi bs)\), \(I_4+I_5\) can be merged as follows:
\begin{equation}
	\begin{small}
		\begin{aligned}
			I_4+I_5 
			&= 4\sum_{r=1}^{2M}\sum_{s=1}^{2N}\Phi(z,s,N)\Phi(c,r,M)\cos(2\pi ar)\cos(2\pi bs)\\
			&=\frac{I_2I_3}{t}.
		\end{aligned}
	\end{small}
\end{equation}
Finally, (\ref{D2}) can be obtained, and the lemma is proved.

\bibliographystyle{IEEEtran}
\bibliography{ref}

\end{document}